\documentclass[a4paper]{article}
\usepackage[utf8]{inputenc}
\usepackage[a4paper, total={16cm, 22cm}]{geometry}

\usepackage[affil-it]{authblk}

\usepackage{amsmath,amsthm,amssymb,amsfonts,stmaryrd,tensor,scalerel}
\usepackage{graphicx}
\usepackage{xcolor}
\usepackage[all,cmtip]{xy}
\usepackage{placeins}
\usepackage{array}
\usepackage{multirow}

\usepackage{pgf}
\usepackage[absolute, overlay]{textpos}

\textblockorigin{\paperwidth}{0.0 pt}

\newtheorem{proposition}{Proposition}[section]
\newtheorem{corollary}[proposition]{Corollary}

\theoremstyle{plain}
\newtheorem{remark}{Remark}[section]
\theoremstyle{definition}
\newtheorem{definition}{Definition}[section]

\newcommand{\de}{\,\mathrm{d}}
\newcommand{\adj}{\mathrm{adj}}

\newcommand{\Hess}{\textnormal{Hess}}
\newcommand{\hess}{\det\textnormal{Hess}}

\newcommand{\rn}[1]{{\mathbb{R}^#1}}

\newcommand{\cotb}[1]{T^*\mathbb{R}^#1}
\newcommand{\Tr}{\textnormal{Tr\,}}

\newcommand{\MA}{Monge-Ampère }
\newcommand{\LR}{Lychagin-Rubtsov }
\newcommand{\CS}{Chynoweth-Sewell }

\newcommand{\jet}[1]{J^1\mathbb{R}^#1}

\newcommand{\PV}{q_g}

\newcommand{\TTT}[2]{T^{(#1)}_{#2}}

\newcommand{\CVL}[1]{CV_{#1}^L}

\title{Solutions and Singularities of the Semigeostrophic Equations via the Geometry of Lagrangian Submanifolds}
\author{R. D'Onofrio${}^{1,2,3}$, G. Ortenzi${}^{1,2}$, I. Roulstone${}^3$ and V. Rubtsov${}^{4,5}$  }
\date{\today}

\affil{{\small  $^1$Dipartimento di Matematica e Applicazioni, Universit\`a di  Milano-Bicocca, \\ Via Roberto Cozzi 55, I-20125 Milano, 
Italy
}\\
{\small  r.donofrio1@campus.unimib.it, giovanni.ortenzi@unimib.it}
\medskip\\
{\small ${}^2$
INFN, Sezione di Milano-Bicocca, Piazza della Scienza 3, 20126 Milano, Italy}
\medskip\\
{\small  $^3$Department of Mathematics, University of Surrey, Guildford GU2 7XH, UK
}\\
{\small  r.d'onofrio@surrey.ac.uk, i.roulstone@surrey.ac.uk}
\medskip\\
{\small ${}^4$Univ.Angers, CNRS, LAREMA, SFR MATHSTIC, F-49000 Angers, France}
\\
{\small volodya@univ-angers.fr}
\medskip\\
{\small ${}^{5}$IGAP (Institute for Geometry and Physics), Trieste, Italy}

}

\begin{document}

\maketitle

\pgfmathwidth{"top right corner at page \thepage"}
\begin{textblock}{\pgfmathresult}[1, 0](0, 1)
\noindent
{\fontfamily{cmr}\selectfont
DMUS-MP-22/19}
\end{textblock}

\begin{abstract}
    Using \MA geometry, we study the singular structure of a class of nonlinear \MA equations in three dimensions, arising in geophysical fluid dynamics. We extend seminal earlier work on \MA geometry by examining the role of an induced metric on Lagrangian submanifolds of the cotangent bundle. In particular, we show that the signature of the metric serves as a classification of the \MA equation,  while singularities and elliptic-hyperbolic transitions are revealed by degeneracies of the metric. The theory is illustrated by application to an example solution of the semigeostrophic equations.
\end{abstract}

\section{Introduction}


Atmospheric fronts are a salient feature of mid-latitude weather systems. From the viewpoint of mathematical modelling, fronts are understood as material interfaces, advected by the fluid flow, across which the physical features undergo a jump discontinuity.
One of the most successful approaches to mathematical modelling of weather fronts is Hoskins' semigeostrophic (SG) equations \cite{CP84,HB72,Hos75}.
In \cite{CS89}, Chynoweth and Sewell recognised the presence of a Legendre duality structure between four different sets of variables in SG theory akin to the classical quartet of dual potentials in thermodynamics. In the same paper, Chynoweth and Sewell showed how singularities of the Legendre mapping could be used to model flows containing a weather front. Their approach is reminiscent of the studies of shock waves in stationary gas flows, where it is known under the name of hodograph transformation (see for example Chapter 12 of \cite{LL87}).

Semigeostrophic flows are completely described by a single function called the geopotential. Moreover, they conserve a form of Ertel's potential vorticity, which in turn is related to the geopotential by a \MA type equation. Denoting the geopotential by $P(x,y,z,t)$ and the potential vorticity by $\PV(x,y,z,t)$, we may write this relation as
\begin{equation}
\label{CS P matrix}
    \begin{vmatrix}
    P_{xx} & P_{xy} & P_{xz}\\
    P_{yx} & P_{yy} & P_{yz}\\
    P_{zx} & P_{zy} & P_{zz}
    \end{vmatrix}
    =C\PV,
\end{equation}
where $C$ is a constant depending on the physical parameters of the model, and subscripts denote partial differentiation (see for example \cite{SC87}).
Time dependence is implicit in (\ref{CS P matrix}) as no time derivatives are involved. Therefore, (\ref{CS P matrix}) represents a kinematic constraint between the geopotential and the potential vorticity, and can be studied by assuming time as a fixed parameter. This approach was also used in \cite{CS89}, where the authors provide several examples of couples $(P,\PV)$ satisfying (the 2D version of) (\ref{CS P matrix}) and capable of modelling atmospheric fronts frozen in time. Both the kinematic and the dynamic view are taken in \cite{HS90}.

The kinematic approach to singularities can be studied using the geometrical framework of \MA equations pioneered by Lychagin and his school (see for example \cite{KLR2006}).
Delahaies and Roulstone \cite{DR2010} have explored the implications \MA geometry for the shallow water version of the SG equations, while Roulstone \textit{et al.} \cite{RBG2009} and Banos \textit{et al.} \cite{BRR2016} have conducted similar studies for the inviscid Navier-Stokes (Euler) equations.
This paper follows this line of research, and investigates the relevance of \MA geometry to the study of singularities of the incompressible 3D SG equations.
One of the main advantages of the geometric approach to \MA equations is a clear and intuitive understanding of  the notion of a generalized solution. While a classical solution is a function $P(x,y,z,t)$, a generalized solution is a Lagrangian submanifold $L$ in the cotangent bundle $\cotb{3}$ (the phase space) of the physical space $\rn{3}$. One thinks of the manifold $L$ as the multivalued graph of the gradient of the geopotential, understood as a map $\nabla P:\rn{3}\to\rn{3}$, in $\cotb{3}\simeq \rn{3}\times\rn{3}$. To recover physical information, a generalized solution $L$ must be projected onto the physical space and singularities can arise in the process. This geometrical perspective on solutions and their singularities was first introduced by Vinogradov and Kupershmidt in their work on Hamilton-Jacobi theory (§8 of \cite{VK77}). Kossowski \cite{Kos91} has independently proposed the same formalism for studying singularities of symplectic \MA equations in two independent variables.
A similar viewpoint is adopted by Ishikawa and Machida \cite{IM2006a,IM2006b} for classifying generic singularities of Hessian type \MA equations in two and three independent variables.
We refer to \cite{Lyc85} for further examples of application of this formalism to more general nonlinear PDEs.

In this work, we present an alternative approach to singularities based on pseudo-Riemannian geometry. For classifying symplectic \MA equations in three independent variables, Lychagin and Rubtsov \cite{LR83} introduced a metric tensor $g_\omega$ on $\cotb{3}$ (formula (\ref{LR metric 3d}) below) and showed that its signature distinguishes the various classes -- elliptic, hyperbolic, and parabolic (see also \cite{Ban2002}). For the particular case of equation (\ref{CS P matrix}), this metric has signature (3,3), and gives $\cotb{3}$ the structure of a pseudo-Riemannian manifold. 

Every generalized solution $L\subset\cotb{3}$ inherits a metric structure from the ambient space associated with the pull-back metric $h_\omega:=g_\omega|_L$ and this constitutes the main focus of the present work. We study the properties of singular solutions to (\ref{CS P matrix}) through pseudo-Riemannian geometry of Lagrangian submanifolds.
Our main goal is to provide an understanding of the pull-back \LR metric $h_\omega$ from the viewpoint of PDE theory. We claim that $h_\omega$ metric plays the same role for \MA equations as the coefficient matrix does for linear second order PDEs. This means that it can be used to define the symbol type of the underlying \MA equation, and, in hyperbolic regime, to construct the characteristic surfaces. In our general setting, the \MA equations under consideration may be of mixed type. We show that $h_\omega$ is Riemannian on elliptic branches of $L$ and Lorentzian on hyperbolic ones. Moreover, we prove that elliptic-hyperbolic transitions and kinematic singularities coincide for equation (\ref{CS P matrix}), implying that $h_\omega$ degenerates on the singular locus of $L$. In this sense, we claim that the pull-back metric is a diagnostic tool for studying singularities.

In Section \ref{sec background}, we give some background on SG equations and \MA geometry. We state and prove our results about the pull-back of the \LR metric and the symbol type of the \MA equation in Section \ref{sec: symbol}. Finally, we present an explicit example of a solution to the SG equations, illustrating the aforementioned results, in Section \ref{sec:examples}.

\section{Background and methods}
\label{sec background}
In this review section we recall some basics about the semigeostrophic system and the geometry of \MA equations.
\subsection{Semigeostrophic equations}
Hoskins' semigeostrophic equations \cite{Hos75} are an approximation to the Euler system of fluid dynamics intended to model large scale motion of the atmosphere. They are usually written
\begin{equation}
\label{SG system}
    \begin{cases}
    \displaystyle{\frac{D u_g}{D t}-fv+\frac{\partial \phi}{\partial x}=0,\quad \frac{D v_g}{D t}+fu+\frac{\partial \phi}{\partial y}=0,\quad \frac{g\theta}{\theta_0}=\frac{\partial \phi}{\partial z},} & (momentum)\\
   \displaystyle{ \frac{\partial u}{\partial x}+\frac{\partial v}{\partial y}+\frac{\partial w}{\partial z}=0,} & (mass)\\
    \displaystyle{\frac{D\theta}{D t}=0,} & (energy)
    \end{cases}
\end{equation}
where
\begin{equation}
\label{material derivative}
    \frac{D}{Dt}=\frac{\partial}{\partial t}+u\frac{\partial}{\partial x}+v\frac{\partial}{\partial y}+w\frac{\partial}{\partial z},
\end{equation}
is the material time derivative, and $\{x,y,z\}$ form a Cartesian coordinate system with $y$ directed pole-ward and $z$ directed vertically. The unknowns are the fluid velocity field $(u,v,w)$, the geopotential $\phi$, and the potential temperature $\theta$ (see \cite{CS89} for a detailed definition of these variables). The positive constants $f\approx 10^{-4}$Hz and $g\approx 10\textnormal{m/s}^2$ account for the effects of Earth's rotation and gravity, while $\theta_0$ is a reference value for $\theta$. Further, there are two main approximations at work in (\ref{SG system}). First, hydrostatic balance is assumed, which results in neglecting the vertical acceleration term in the momentum balance. Second, the fluid flow is supposed to be close to geostrophic equilibrium (see \cite{Hos75} , $\S$3): this is accounted for in (\ref{SG system}) through Hoskins' ``geostrophic momentum approximation'', which replaces the fluid velocity in the horizontal acceleration terms with its geostrophic part,
\begin{equation}
\label{geostrophic wind}
    u_g:=-\frac{1}{f}\frac{\partial\phi}{\partial y},\qquad v_g:=\frac{1}{f}\frac{\partial\phi}{\partial x}.
\end{equation}
System (\ref{SG system}) implies an important conservation property, namely, the conservation of Ertel's potential vorticity along fluid trajectories,
\begin{equation}
    \frac{D \PV}{D t}=0.
\end{equation}
This quantity represents the projection of the absolute geostrophic vorticity,
\begin{equation}
\label{absolute vorticity}
    \zeta_g=
    \left(-\frac{\partial v_g}{\partial z}+\frac{1}{f}\frac{\partial(u_g,v_g)}{\partial(y,z)},\frac{\partial u_g}{\partial z}+\frac{1}{f}\frac{\partial(u_g,v_g)}{\partial(y,z)},f+\frac{\partial v_g}{\partial x}-\frac{\partial u_g}{\partial y}+\frac{1}{f}\frac{\partial(u_g,v_g)}{\partial(x,y)}\right),
\end{equation}
along the gradient of the potential temperature,
\begin{equation}
\label{PV}
    \PV:=\zeta_g\cdot\nabla \theta.
\end{equation}
Notice that equation (\ref{PV}) plus (\ref{geostrophic wind}) and the hydrostatic balance condition (\ref{SG system}) provides a direct link between the potential vorticity and the geopotential, which comes in the form of a \MA equation. This statement is made more clear by introducing the modified geopotential,
\begin{equation}
\label{modified geopotential P}
    P:=\frac{\phi}{f^2}+\frac{x^2}{2}+\frac{y^2}{2},
\end{equation}
which allows one to write equation (\ref{PV}) as
\begin{equation}
\label{CS P dimensional}
    \PV=\frac{f^3\theta_0}{g}\hess(P),
\end{equation}
where $\hess(P)$ denotes the determinant of the Hessian matrix of $P$ with respect to the spatial variables. In this work, we are exclusively interested in kinematic aspects of system (\ref{SG system}) as encoded in the \MA equation (\ref{CS P dimensional}). 

The mathematical structure of the semigeostrophic equations appears even more clearly on introducing dimensionless variables. Following \cite{Oli2006}, we write the semigeostrophic system in dimensionless form as
\begin{gather}
\label{SG nondimensional}
\begin{split}
    \epsilon\frac{D u_g}{D t}-v+\frac{\partial\phi}{\partial x}=0,\qquad \epsilon\frac{D v_g}{D t}+u+\frac{\partial \phi}{\partial y}=0\qquad \frac{D\theta}{D t}=0,\\
    \frac{\partial u}{\partial x}+\frac{\partial v}{\partial y}+\frac{\partial w}{\partial z}=0,\qquad u_g=-\frac{\partial \phi}{\partial y},\qquad v_g=\frac{\partial \phi}{\partial x},\qquad \theta=\frac{\partial \phi}{\partial z}.
\end{split}
\end{gather}
The dimensionless parameter $\epsilon$ is the Rossby number,
\begin{equation}
    \epsilon=\frac{U}{fL},
\end{equation}
which involves the typical horizontal length and velocity scales, and represents the ratio between inertial and Coriolis forces. The value of $\epsilon$ is typically $\approx 0.1$ in semigeostrophic flows. 
A consistent dimensionless expression for the modified geopotential is
\begin{equation}
    P=\frac{x^2}{2}+\frac{y^2}{2}+\epsilon\phi,
\end{equation}
and the dimensionless version of (\ref{CS P dimensional}) reads accordingly
\begin{equation}
\label{CS P}
    \hess(P)=\epsilon \PV.
\end{equation}

\begin{remark}
\label{rem:positive PV}
Shutts and Cullen \cite{SC87} used equation (\ref{CS P}) to study stability of semigeostrophic flows with respect to small displacement of fluid parcels. They found that a necessary condition for parcel stability is the (spatial) convexity of the geopotential $P$, which implies strict positivity of the potential vorticity.
Although we place no a priori hypothesis on the convexity of $P$, we shall always assume $\PV>0$ henceforth.
\end{remark}

Much of the current interest in the semigeostrophic equations is motivated by a change of variable due to Hoskins \cite{Hos75}, known as the ``geostrophic momentum transformation'', which has drastically improved the general comprehension of semigeostrophic flows. As later recognised by Chynoweth and Sewell \cite{CS89}, this change of variable may be interpreted as a Legendre type transformation, and we adopt their perspective for describing it. We start by introducing the (dimensionless) horizontal components of the absolute momentum,
\begin{equation}
    M=x+\epsilon v_g=x+\epsilon\frac{\partial\phi}{\partial x},\qquad N=y-\epsilon u_g=y+\epsilon \frac{\partial\phi}{\partial y},
\end{equation}
which stands in a special relation with the geopotential,
\begin{equation}
\label{M N P}
    \frac{\partial P}{\partial x}=M,\qquad \frac{\partial P}{\partial y}=N.
\end{equation}
Moreover, the hydrostatic balance condition is written in terms of $P$ as
\begin{equation}
\label{theta P}
    \frac{\partial P}{\partial z}=\epsilon\frac{\partial \phi}{\partial z}=\epsilon\theta.
\end{equation}
Equations (\ref{M N P}) and (\ref{theta P}) are the starting point for \cite{CS89}, where a quartet of Legendre transformations of $P$ is identified. In fact, these relations open the way to the geometrization of the kinematic equation (\ref{CS P}). For the remainder of this section, we will frame the work of Chynoweth and Sewell within the geometrical theory of \MA equations (see for example \cite{KLR2006}).
For convenience of exposition, we will use the notation,
\begin{equation}
\label{dual variables}
    (X,Y,Z):=(M,N,\epsilon\theta)=\nabla P.
\end{equation}

\subsection{\MA structure and Legendre duality}
In this work, we only deal with symplectic \MA equations in three independent variables. Here, the word \textit{symplectic} means that the equation's coefficients can only depend on the independent variables and the gradient of the dependent variable, but not on the value itself of the dependent variable. This is true for equation (\ref{CS P}) if we understand the potential vorticity $\PV$ as a function of space (and possibly time). As the name suggests, symplectic \MA equations are associated with the symplectic geometry of a manifold, the phase space, representing the cotangent bundle to the space of independent variables. In the particular case of equation (\ref{CS P}), the associated symplectic manifold is $\cotb{3}$, which we endow with coordinates $(x,y,z,X,Y,Z)$ and the canonical symplectic form
\begin{equation}
\label{canonical symplectic}
    \Omega=\de X\wedge\de x+\de Y\wedge\de y+\de Z\wedge\de z.
\end{equation}
A function $P:\rn{3}\to \mathbb{R}$ induces a section of the cotangent bundle through its differential, $dP:\rn{3}\to\cotb{3}$. Note that the image of $dP$ in $\cotb{3}\simeq \rn{3}\times\rn{3}$ coincides with the graph of the gradient $\nabla P:\rn{3}\to\rn{3}$. For any given 3-form $\omega$ on $\cotb{3}$, we can define a map
\begin{equation}
    \Delta_\omega:C^\infty(\rn{3})\to\Omega^3(\rn{3}),
\end{equation}
taking functions to 3-forms on $\rn{3}$, which we call the \MA operator associated to $\omega$. It associates a function with the restriction of $\omega$ to its graph,
\begin{equation}
    \Delta_\omega(P):=(dP)^*\omega,
\end{equation}
where the superscript ${}^*$ denotes the pull-back.
The correspondence between \MA operators and 3-forms is not 1-to-1 (several forms produce the same operator), but it can be made so by taking a suitable quotient of the space of 3-forms on $\cotb{3}$.
Not every 3-form produces a nonzero \MA operator and those which do are called \textit{effective}. This induces an equivalence relation on the space of 3-forms on $\cotb{3}$ -- two forms are equivalent if they differ by a non-effective form. Thus, an equivalence class of 3-forms gives rise to one and the same \MA operator. We will exclusively deal with effective forms and shall make no explicit distinction between an effective form and the class it represents. In the symplectic case, effective 3-forms can be characterized as those $\omega\in\Omega^3(\cotb{3})$ which satisfy $\omega\wedge\Omega=0$. Its easy to verify that the effective 3-form on $\cotb{3}$ associated to equation (\ref{CS P}) is
\begin{equation}
\label{effective form}
    \omega=\de X\wedge\de Y\wedge \de Z-\epsilon\PV\de x\wedge\de y\wedge\de z.
\end{equation}

The equation $\Delta_\omega=0$ is called the \MA equation associated to $\omega$, and we denote it by $E_\omega$. The \MA equation corresponding to (\ref{effective form}) is, in Cartesian coordinates, (\ref{CS P}), as can be immediately verified by direct calculation.
The next step in the definition of this geometrical framework is the notion of a solution.
A \textit{generalized} solution to $E_\omega$ is a smooth Lagrangian submanifold\footnote{A Lagrangian submanifold $L$ in a symplectic manifold $(M^{2n},\Omega)$ is an isotropic submanifold ($\Omega|_L=0$) of the maximum possible dimension ($\textnormal{dim}(L)=n$).} $L\in\cotb{3}$ such that $\omega|_L=0$. A \textit{classical} solution is one which is globally represented as the graph of a function, meaning that $L$ is the range of the differential $dP$ for some twice differentiable function $P(x,y,z)$. Note that in this latter case the condition $\omega|_L=0$ reads
\begin{equation}
\label{CS P geometric}
    \Delta_\omega(P)=0.
\end{equation}
Generalized solutions are precisely those which can be locally (but not globally) represented as the graph of a function near generic points. 
Stated differently, the mapping $\pi_L:=\pi|_L$, where
\begin{equation}
\label{Lagrangian projection L}
    \pi:\cotb{3}\to \mathbb{R}^3,\qquad (x,y,z,X,Y,Z)\mapsto (x,y,z),
\end{equation}
represents the canonical bundle projection, is a diffeomorphism between $L$ and $\rn{3}$ as long as $L$ is the graph of a classical solution. Otherwise, we must distinguish between \textit{regular} points, where $\de \pi_L$ has maximal rank, and \textit{singular} points, where it has not. In this latter case, the diffeomorphism property is local, and only holds near regular points.

Although generalized solutions cannot be globally represented as the graph of a function, they admit an alternative local representation in terms of a single function near any point (either regular or singular). Note that
\begin{equation}
\label{tautological 1-form}
    \Omega=\de\alpha,
\end{equation}
where $\alpha:=X\de x+Y\de y+Z\de z$ is the tautological 1-form on $\cotb{3}$. If $L$ is a Lagrangian submanifold, then $\alpha|_L$ is closed, and so there exist a function $f$ on $L$ such that, locally,
\begin{equation}
\label{generating function f}
    \alpha|_L=df.
\end{equation}
We call $f$ a generating function for the Lagrangian submanifold $L$. Once a coordinate set on $L$ is selected, (\ref{generating function f}) reduces to an algebraic system of equations whose zero set in $\cotb{3}$ identifies the submanifold $L$. In the neighbourhood of any given point, coordinates on $L$ can always be chosen as a suitable 3-subset of the cotangent coordinates. Overall, there are $2^3=8$ possible choices to pick a 3-subset from $\{x,y,z,X,Y,Z\}$ and, therefore, as many classes of generating functions.
The Legendre dual potentials of Chynoweth and Sewell \cite{CS89}, 
\begin{equation}
\label{CS potentials}
    R(X,Y,Z), \qquad S(X,Y,z), \qquad T(x,y,Z),
\end{equation}
provide some physically relevant examples of generating functions. We explicitly work out the description of a Lagrangian submanifold $L$ in terms of $S$. Setting
\begin{equation}
    f(X,Y,z)=Xx+Yy-S(X,Y,z)
\end{equation}
in equation (\ref{generating function f}), gives,
\begin{equation}
    Z\de z=x\de X+y\de Y-\frac{\partial S}{\partial X}\de X-\frac{\partial S}{\partial Y}\de Y-\frac{\partial S}{\partial z}\de z,
\end{equation}
which in turn implies
\begin{equation}
\label{L S}
    x=\frac{\partial S}{\partial X}(X,Y,z),\qquad y=\frac{\partial S}{\partial Y}(X,Y,z), \qquad Z=-\frac{\partial S}{\partial z}(X,Y,z).
\end{equation}
The combined zero set of equations (\ref{L S}) in $\cotb{3}$ identifies the Lagrangian submanifold $L$ generated by $S$. Similarly, the choices
\begin{equation}
    f(X,Y,Z)=Xx+Yy+Zz-R(X,Y,Z)
\end{equation}
and
\begin{equation}
    f(x,y,Z)=Zz+T(x,y,Z)
\end{equation}
lead to a local description of as many classes of Lagrangian submanifolds as the zero set of, respectively,
\begin{equation}
\label{L R}
    x=\frac{\partial R}{\partial X}(X,Y,Z),\qquad y=\frac{\partial R}{\partial N}(X,Y,Z), \qquad z=\frac{\partial R}{\partial Z}(X,Y,Z),
\end{equation}
and
\begin{equation}
\label{L T}
    z=-\frac{\partial T}{\partial Z}(x,y,Z),\qquad X=\frac{\partial T}{\partial x}(x,y,Z),\qquad Y=\frac{\partial T}{\partial y}(x,y,Z).
\end{equation}
\begin{remark}
In \cite{CS89}, the term ``dual space'' is used to refer to the Cartesian space of coordinates $(X,Y,z)$, $(x,y,Z)$ or $(X,Y,Z)$. The geometrical setting brings out the true nature of the dual space as the local coordinate representation of a Lagrangian submanifold.
\end{remark}
The condition that $L$ is a solution to the \MA equation results in a condition on the generating function itself which again takes the form of a \MA equation.
For the Chynoweth--Sewell potentials (\ref{CS potentials}), this condition respectively reads
\begin{align}
\label{CS R}
    &\hess(R)=\frac{1}{\epsilon q_g},\\
    \label{CS S}
    &\epsilon\PV\left[\frac{\partial^2S}{\partial X^2}\frac{\partial^2S}{\partial Y^2}-\left(\frac{\partial^2S}{\partial X\partial Y}\right)^2\right]+\frac{\partial^2S}{\partial z^2}=0,\\
    \label{CS T}
    &\frac{\partial^2T}{\partial x^2}\frac{\partial^2T}{\partial y^2}-\left(\frac{\partial^2T}{\partial x\partial y}\right)^2+\epsilon\PV\frac{\partial^2T}{\partial Z^2}=0.
\end{align}
We close this section by noting that $R(X,Y,Z)$ plays a distinguished role among the Chynoweth-Sewell potentials. Indeed, if $(x,y,z)$ are entirely replaced by $(X,Y,Z)$ in the role of independent variables, the whole semigeostrophic system (\ref{SG system}) reduces to just two equations (see for example \cite{RN2006}), and comprises the \MA equation (\ref{CS R}) and a transport equation for the potential vorticity,
\begin{equation}
\label{SG dual variables}
    \begin{cases}
    \displaystyle{\epsilon\PV\hess(R)=1},\\
    \displaystyle{\frac{\partial \PV}{\partial t}=\frac{\partial(q_g,\Psi)}{\partial(X,Y)}},\\
    \displaystyle{\Psi:=\frac{X^2}{2}+\frac{Y^2}{2}-R}.
    \end{cases}
\end{equation}
System (\ref{SG dual variables}) provides a clear distinction between the model's kinematics, encoded in the \MA equation, and its dynamics, represented by the transport of vorticity.

\section{Pseudo-Riemannian geometry and classification of nonlinear PDEs}
\label{sec: symbol}

To classify\footnote{Two effective 3-forms $\omega_1$ and $\omega_2$ on $\cotb{3}$ are locally equivalent if there is a local symplectomorphism of the phase space pulling $\omega_2$ back to $\omega_1$.} symplectic 3-dimensional \MA operators, Lychagin and Rubtsov \cite{LR83} introduced a symplectic invariant attached to any given effective 3-form on $\cotb{3}$. It may be defined through the relation (see \cite{BRR2016})
\begin{equation}
\label{LR metric 3d}
    g_\omega(\xi_1,\xi_2)\frac{\Omega^3}{3!}=\iota_{\xi_1} \omega\wedge \iota_{\xi_2}\omega\wedge \Omega,
\end{equation}
which holds for each pair of vector fields $\xi_1,\xi_2$ on $\cotb{3}$. Note that $g_\omega$ is a symmetric bilinear form on the phase space $\cotb{3}$. Moreover, as shown in \cite{Ban2002}, $g_\omega$ happens to be non-degenerate for certain classes of \MA operators. When this holds, $g_\omega$ defines a Riemannian (or pseudo-Riemannian) metric on the phase space $\cotb{3}$, which we call the \LR metric.
For $\omega$ and $\Omega$ given by (\ref{effective form}) and (\ref{canonical symplectic}), (\ref{LR metric 3d}) yields
\begin{equation}
\label{LR metric SG}
    g_\omega=2\epsilon\PV(\de x\de X+\de y\de Y+\de z\de Z).
\end{equation}
Thus, as long as the \MA equation (\ref{CS P}) is concerned, $g_\omega$ is a pseudo-Riemannian metric with signature $(3,3)$ over the phase space $\cotb{3}$. 
The \LR metric induces a pseudo-metric on every submanifold of $\cotb{3}$, and, in particular, on solutions of the \MA equation. Let
\begin{equation}
    \iota:L\to\cotb{3}
\end{equation}
be a generalized solution to (\ref{CS P}), that is, a Lagrangian submanifold such that $\iota^*\omega=0$. Then, $L$ inherits the (pseudo)-Riemannian structure of ambient space as provided by the pull-back metric
\begin{equation}
\label{pull-back metric}
    h_\omega:=\iota^*g_\omega.
\end{equation}
\begin{remark}
The \LR metric (\ref{LR metric SG}) is pseudo-Riemannian (non-degenerate) as long as $\PV\ne 0$. This is always true in the present work as we assume $\PV>0$ (see Remark \ref{rem:positive PV}). However, nothing can be said a priori about the pull-back metric (\ref{pull-back metric}) which depends on the solution $L$ and the position on $L$. In the general case, $h_\omega$ can be either Riemannian, pseudo-Riemannian and even degenerate.
We call $h_\omega$ degenerate at a point $e\in L$ if there exists a tangent vector $\xi_1\in T_e L$ such that $h_\omega(\xi_1,\xi_2)=0$ for any $\xi_2\in T_e L$.
\end{remark}

We are now able to describe one of the main results of this paper: the characterization of MAEs in terms of the geometry of $L$. We show that there is a correspondence between the signature of (\ref{pull-back metric}) and the symbol type (elliptic/parabolic/hyperbolic) of the MAE (\ref{CS P}).
This is made precise in Proposition \ref{prop: h vs type} and leads to a natural characterization of the equation type in terms of $h_\omega$ (Definition \ref{def: type vs h}).
We start with classical solutions, amenable for treatment through linearization. Next, we consider generalized solutions with the aid of generating functions. Then we draw these ideas together in §\ref{sub:singularities}, stating the relationship between elliptic-hyperbolic transitions, projection singularities and the pull-back of the \LR metric. The section ends with an account on characteristic surfaces in terms of the pull-back metric.

\subsection{Classical solutions}
We begin with recalling a classical definition from PDE theory (see for example \cite{CH62}):
\begin{definition}[Type of a linear equation]
A second order linear PDE with principal part
\begin{equation}
    \sum_{i,j=1}^n a_{ij}(x)\frac{\partial^2 u}{\partial x_i\partial x_j}
\end{equation}
is called \textit{elliptic} if the eigenvalues of the symmetric matrix $A=[a_{ij}]$ have the same sign; \textit{hyperbolic} if one eigenvalue has the opposite sign from the others; and \textit{parabolic} if there is at least one zero eigenvalue.
\end{definition}
The notion of an equation type has been generalized to nonlinear equations by Harvey and Lawson \cite{HL82} as follows (see also §2 of \cite{Duz2004} for a geometrical perspective).
\begin{definition}
\label{def HL}
The type of a nonlinear equation \textit{at} a given solution is the type of its linearization about the solution.
\end{definition}
We are thus led to consider the linearization of equation (\ref{CS P}) about a fixed solution. Let $P+\delta P$ be a perturbation of some exact solution $P$ to (\ref{CS P}) with $\delta P=o(\epsilon)$. Introducing this ansatz into equation (\ref{CS P}) and using the Jacobi formula for determinants leads, to the first order in $\delta P$, to the linear equation satisfied by the perturbation field,
\begin{equation}
\label{linearized P}
    \Tr\left[\adj(\Hess(P))\Hess(\delta P)\right]=0,
\end{equation}
with coefficient matrix
\begin{equation}
\label{linearization matrix}
    A=\adj(\Hess(P)),
\end{equation}
where ``$\adj$'' denotes the adjugate matrix.

The assumption of strictly positive potential vorticity (Remark \ref{rem:positive PV}) implies that (\ref{linearized P}) is elliptic if $P$ is (spatially) convex and hyperbolic if $P$ is saddle shaped. Definition \ref{def HL} allows us to bring this information to the nonlinear equation (\ref{CS P}) as it stands. Also note that equation (\ref{CS P}) is nowhere parabolic as long as classical solutions are considered. In fact, equation (\ref{CS P}) with $q_g > 0$ itself prevents the eigenvalues of $\Hess(P)$ (and thus those of $A$) from vanishing. We are now in a position to prove the following
\begin{proposition}
Let $P$ be a classical solution to (\ref{CS P}) and let $L\subset \cotb{3}$ denote the graph of $dP$.
Then, the pull-back \LR metric on $L$,
\begin{equation}
\label{h P0 def}
    h_\omega:=(dP)^*g_\omega,
\end{equation}
has matrix representation
\begin{equation}
\label{LR vs A}
    h_\omega=2\adj(A)=2\epsilon\PV\Hess(P),
\end{equation}
where $A$ is the linearization matrix (\ref{linearization matrix}). Moreover, $h_\omega$ has signature $(3,0)$ if equation (\ref{CS P}) is elliptic at the solution $P$ and signature $(1,2)$ if (\ref{CS P}) is hyperbolic at $P$.
\end{proposition}
\begin{proof}
By direct calculation. Recalling that,
\begin{equation}
    dP:(x,y,z)\mapsto \left(x,y,z,\frac{\partial P}{\partial x},\frac{\partial P}{\partial y},\frac{\partial P}{\partial z}\right),
\end{equation}
equation (\ref{h P0 def}) implies
\begin{equation}
    h_\omega=(dP)^*g_\omega=2\epsilon\PV\frac{\partial^2 P}{\partial x^i\partial x^j}\de x^i\de x^j,
\end{equation}
where we have set $(x^1,x^2,x^3)\equiv(x,y,z)$, and summation on repeated indices is implied. Therefore, $h_\omega$ has matrix representation
\begin{equation}
    h_\omega=2\epsilon\PV\Hess(P).
\end{equation}
On the other hand, equation (\ref{linearization matrix}) plus the algebraic identity
\begin{equation}
    \adj(\adj(M))=\det(M)^{n-2}M,
\end{equation}
holding for any square matrix $M\in \mathbb{R}^{n\times n}$ with $n>2$, gives
\begin{equation}
    \adj(A)=\adj(\adj(\Hess(P)))=\det(\Hess(P))\Hess(P)=\epsilon\PV\Hess(P),
\end{equation}
and thus equation (\ref{LR vs A}) follows.
As for the second part, we observe that the determinant,
\begin{equation}
    \det(h_\omega)=8\det(A)^2=8\det\Hess(P)^4=8(\epsilon\PV)^4,
\end{equation}
is always positive, so the eigenvalues of (\ref{LR vs A}) can only be \textit{(i)} all positive or \textit{(ii)} one positive and two negative. According to (\ref{LR vs A}), \textit{(i)} occurs when $P$ is convex and case \textit{(ii)} occurs when $P$ is saddle shaped. Therefore, cases \textit{(i)} and \textit{(ii)} correspond to (\ref{CS P}) being respectively elliptic or hyperbolic.
\end{proof}
\begin{remark}
The first equality in (\ref{LR vs A}) turns out to be a general property of symplectic 3-dimensional \MA equations. In particular, it applies to the dual equations (\ref{CS R}), (\ref{CS S}) and (\ref{CS T}) above.
\end{remark}
Definition \ref{def HL} is no longer directly applicable when generalized solutions are allowed, as the whole linearization process is ill defined. However, equation (\ref{LR vs A}) suggests a characterization of ellipticity of equation (\ref{CS P}) based on $h_\omega$, which applies to generalized solutions too. Indeed, the equation type at a generalized solution is directly traceable to the signature of $h_\omega$ in a definite way.
In the remainder of this section, we will prove consistency of this characterization by relying on the local description of generalized solutions in terms of generating functions.

\subsection{Generalized solutions}
We recall that any generating function $f$ of a generalized solution $L$ satisfies a \MA equation which arises from expressing the condition $\omega|_L=0$ in local coordinates on $L$. Moreover, as $f$ is a classical solution to this equation, there are no obstructions to linearization. Thus, we can give the following
\begin{definition}
\label{def: f vs type}
Let $L$ be a generalized solution to (\ref{CS P}) locally generated by a generating function $f$. We say that (\ref{CS P}) is elliptic, parabolic or hyperbolic at some point $e\in L$ if $f$ satisfies a \MA equation of the same type at the point. 
\end{definition}
We remark that the symbol type of a differential equation is invariant under a change of variables \cite{CH62}, and this ensures consistency of Definition \ref{def: f vs type}. We are thus in a position to prove the
\begin{proposition}
\label{prop: h vs type}
Let $\iota:L\mapsto\cotb{3}$ be a generalized solution to (\ref{CS P}). Then, the pull-back metric $h_\omega=\iota^*g_\omega$ has signature $(3,0)$ on elliptic branches of $L$, $(1,2)$ on hyperbolic branches, and degenerates along parabolic branches.
\end{proposition}
\begin{proof}
This proposition is proved by direct inspection of the linearized \MA equation satisfied by the generating function $f$.
We explicitly carry out the calculations for the case of $f=Zz+T$ (the remaining cases are addressed similarly and lead to the same conclusions).
Thus, let $L$ be some generalized solution locally described by $T(x,y,Z)$ according to (\ref{L T}). Introducing a perturbation $T+\delta T$ of an exact solution $T$ to (\ref{CS T}), with $\delta T=o(\epsilon)$, leads to a linear equation satisfied by the perturbation field,
\begin{equation}
    \frac{\partial^2 T}{\partial x^2}\frac{\partial^2 \delta T}{\partial y^2}+\frac{\partial^2 T}{\partial y^2}\frac{\partial^2 \delta T}{\partial x^2}-2\frac{\partial^2 T}{\partial x\partial y}\frac{\partial^2 \delta T}{\partial x\partial y}+\epsilon\PV\frac{\partial^2 \delta T}{\partial Z^2}=0.
\end{equation}
Its coefficient matrix is
\begin{equation}
\label{principal part T}
    A=\begin{pmatrix}
    \frac{\partial^2 T}{\partial y^2} & -\frac{\partial^2 T}{\partial x\partial y} & 0\\
    -\frac{\partial^2 T}{\partial x\partial y} & \frac{\partial^2 T}{\partial x^2} & 0\\
    0 & 0 & \epsilon\PV
    \end{pmatrix}=
    \begin{pmatrix}
    \adj(H) & 0\\
    0 & \epsilon\PV
    \end{pmatrix},
\end{equation}
where
\begin{equation}
    H:=\begin{pmatrix}
    \frac{\partial^2T}{\partial x^2} & \frac{\partial^2T}{\partial x\partial y}\\
    \frac{\partial^2T}{\partial x\partial y} & \frac{\partial^2T}{\partial y^2}
    \end{pmatrix}.
\end{equation}
On the other hand, the \LR metric $h_\omega=g_\omega|_L$ has the local coordinate expression
\begin{equation}
\label{h T}
    h_\omega=2\epsilon\PV\left(\frac{\partial^2 T}{\partial x^2}\de x^2+2\frac{\partial^2 T}{\partial x\partial y}\de x\de y+\frac{\partial^2 T}{\partial y^2}\de y^2-\frac{\partial^2 T}{\partial Z^2}\de Z^2\right),
\end{equation}
and may be written in matrix form as
\begin{equation}
\label{h T matrix}
    h_\omega=2\epsilon\PV
    \begin{pmatrix}
    \frac{\partial^2T}{\partial x^2} & \frac{\partial^2T}{\partial x\partial y} & 0\\
    \frac{\partial^2T}{\partial x\partial y} & \frac{\partial^2T}{\partial y^2} & 0\\
    0 & 0 & -\frac{\partial^2T}{\partial Z^2}
    \end{pmatrix}=2
    \begin{pmatrix}
    \epsilon\PV H & 0\\
    0 & \det(H)
    \end{pmatrix},
\end{equation}
where the last equality follows from (\ref{CS T}).
We note in passing that
\begin{equation}
\label{h vs A S}
    h_\omega=2\adj(A).
\end{equation}
We see from (\ref{principal part T}) that equation (\ref{CS T}) is elliptic as long as $H$ is positive definite,
which, in light of (\ref{h T matrix}), corresponds to $h_\omega$ having signature $(3,0)$. Parabolic and hyperbolic cases correspond to $\det(H)=0$ and $\det(H)<0$ respectively. Therefore, it follows from equation (\ref{h T matrix}) that $h_\omega$ is degenerate on parabolic branches and of type $(1,2)$ on hyperbolic ones.
\end{proof}
The significance of Proposition \ref{prop: h vs type} is that $h_\omega$ encodes all the essential information about the equation type, and may be used to give an invariant definition of the symbol type based on its signature.  This may be summarized as follows:
\begin{definition}
\label{def: type vs h}
    Let $L$ be a generalized solution to (\ref{CS P}) with \LR metric $h_\omega$. We say that (\ref{CS P}) is elliptic or hyperbolic at a point $e\in L$ if $(h_\omega)_e$ is respectively of type $(3,0)$ or $(1,2)$. We say that (\ref{CS P}) is parabolic at $e\in L$ if $(h_\omega)_e$ is degenerate.
\end{definition}

\subsection{Singularities}
\label{sub:singularities}
In this section we show that elliptic-hyperbolic transitions in (\ref{CS P}) can only occur along through some singularity. This result is closely related to the assumption of strictly positive potential vorticity, which, according to (\ref{CS P}), prevent the eigenvalues of $\Hess(P)$ from changing sign as long as classical solutions are considered. As a result, the pull-back of the \LR metric degenerates on singularities.

We recall that points on $L$ are called \textit{regular} if the tangent map $\de \pi_L$ is surjective, and \textit{singular} otherwise. We denote by $\Sigma L$ the set of singular points on $L$. 
\begin{proposition}
\label{prop: Sigma L = parabolic}
Let $L$ be a generalized solution to (\ref{CS P}). Then the set of parabolic points on $L$ coincides with the singular locus $\Sigma L$.
\end{proposition}
\begin{proof}
Once again, we rely on a local description in coordinates and generating functions to prove our result.
Let a solution $L$ be locally generated by a function $f=Zz+T$, that is,
\begin{equation}
    L=\left\{(x,y,z,X,Y,Z)\in \cotb{3}:X=\frac{\partial T}{\partial x},\;Y=\frac{\partial T}{\partial y},\;z=-\frac{\partial T}{\partial Z}\right\}.
\end{equation}
In local coordinates $\{x,y,Z\}$ on $L$, the projection mapping reads
\begin{equation}
    \pi_L(x,y,Z)=\left(x,y,-\frac{\partial T}{\partial Z}\right),
\end{equation}
and so it is singular on points satisfying
\begin{equation}
\label{Sigma L T}
    \det(\de\pi_L)=-\frac{\partial^2T}{\partial Z^2}=(\epsilon\PV)^{-1}\bigg(\frac{\partial^2 T}{\partial x^2}\frac{\partial^2 T}{\partial y^2}-\left(\frac{\partial^2 T}{\partial x\partial y}\right)^2\bigg)=0.
\end{equation}
On the other hand, we know from the proof of Proposition \ref{prop: h vs type} that parabolic points on $L$ satisfy the same equation.
To complete the proof, one should examine in turn each of the remaining classes of generating functions. However, calculations are almost identical to those we have already exhibited, and we omit them for conciseness.
\end{proof}
Propositions \ref{prop: h vs type} and \ref{prop: Sigma L = parabolic} are combined to
\begin{corollary}
\label{cor: h vs singularity}
Given a generalized solution $L$ to (\ref{CS P}), the induced \LR metric on $L$ degenerates along the singular locus $\Sigma L$.
\end{corollary}
Thus, every regular branch of a multivalued solution $L\subset \cotb{3}$ is of a single type (elliptic or hyperbolic), and transitions are only possible in passing from one branch to another.
\begin{remark}
Chynoweth and Sewell's approach to singularities is by the Legendre transform (see equations (12) of \cite{CS89}).
Once a solution to (\ref{CS R}), (\ref{CS S}) or (\ref{CS T}) is known, the (possibly multivalued) geopotential is recovered by the inverse Legendre transform,
\begin{align}
\label{Legendre transform R}
    &P=Xx+Yy+Zz-R, & &\qquad x=\frac{\partial R}{\partial X},\qquad y=\frac{\partial R}{\partial Y},\qquad z=\frac{\partial R}{\partial Z},\\
\label{Legendre transform S}
    &P=Xx+Yy-S, & &\qquad x=\frac{\partial S}{\partial X},\qquad y=\frac{\partial S}{\partial Y},\\
\label{Legendre transform T}
    &P=Zz+T, & &\qquad z=-\frac{\partial T}{\partial Z}.
\end{align}
Singularities are then identified according to their effects on the graph of the multivalued $P$. Chynoweth and Sewell's viewpoint is reconciled with the geometric viewpoint as follows. We denote by $J^1(\rn{3})\simeq \cotb{3}\times \mathbb{R}$ the bundle of 1-jets over the physical space $\rn{3}(x,y,z)$, and we endow it with coordinates
\begin{equation}
    (x,y,z,u,X,Y,Z).
\end{equation}
In this extended space, a generalized solution is understood as a Legendrian submanifold\;\footnote{A Legendrian submanifold of a contact manifold $(M^{2n+1},C)$ is a $n$-dimensional integral manifold of the contact distribution $C$. If $M=\jet{3}$, $C$ is canonically described as the kernel of the Cartan 1-form $\Theta=du-Xdx-Ydy-Zdz$ (see for example \cite{KLR2006}).} $\hat\iota:L\to\jet{3}$ such that $\hat\iota^*\omega=0$ (here, $\omega$ is understood as a 3-form on $\jet{3}$).
Then, the projection $\hat{\pi}_L=\hat\pi\circ\hat{\iota}$ of $L$ to the base space $J^0\rn{3}$ of the jet bundle, where
\begin{equation}
    \hat{\pi}:\jet{3}\to J^0\rn{3},\quad (x,y,z,u,X,Y,Z)\mapsto (x,y,z,u),
\end{equation}
results in the graph of the (possibly) multivalued geopotential $P$, parametrized by the local coordinates on $L$ by equations (\ref{Legendre transform R}), (\ref{Legendre transform S}) or (\ref{Legendre transform T}). The following commutative diagram summarizes the situation.
\begin{equation*}
    \xymatrix{
    \jet{3} \ar[d]_{\hat{\pi}} & L  \ar@{_{(}->}[l]_{\hat{\iota}} \ar@{^{(}->}[r]^{\iota} \ar[ld]^{\hat{\pi}_L} \ar[rd]_{\pi_L} & \cotb{3} \ar[d]^{\pi}\\
    J^0\rn{3} & {} & \rn{3}
    }
\end{equation*}
\begin{equation*}
    {}
\end{equation*}
\end{remark}

\subsection{Characteristic variety}
The strong connection between $h_\omega$ and the symbol type of (\ref{CS P}) suggests a link with the characteristic surfaces as well.
In this section, we explore the geometry of the light cone of $h_\omega$, and use it to introduce a suitable notion of characteristic surfaces in hyperbolic and parabolic regime.

A central role in this subject is played by vectors of null length, which, borrowing terminology from relativity theory, are called \textit{light-like}. The set of light-like vectors based at a point is called the \textit{light cone} or the \textit{characteristic variety}. This notion is made precise by the following definition, which builds on the work of Kossowski \cite{Kos91} on 2-dimensional \MA equations.
\begin{definition}
Let $g_\omega$ be given by (\ref{LR metric SG}) and let $e\in\cotb{3}$. The cotangent characteristic variety (or simply characteristic variety) $CV_e\subset T_e(\cotb{3)}$ is the cone
\begin{equation}
    CV_e:=\{\xi\in T_e(\cotb{3)}:g_\omega(\xi,\xi)=0\}.
\end{equation}
Let $L\subset\cotb{3}$ be some generalized solution to the (\ref{CS P}) and let $e\in L$. We denote the pull-back of the characteristic variety to $L$ as
\begin{equation}
    \CVL{e}:=\{\xi\in T_eL:h_\omega(\xi,\xi)=0\},
\end{equation}
where $h_\omega=g_\omega|_L$ is the pull-back metric on $L$.
\end{definition}
It easily follows from Definition \ref{def: type vs h} that the characteristic variety $\CVL{e}$ is a full-fledged cone if $e$ is an hyperbolic point, a degenerate cone if $e$ is a parabolic point, and the zero vector, $\{0\}\subset T_eL$, if $e$ is an elliptic point. The characteristic variety is the basic ingredient to build the characteristic surfaces within a generalized solution $L$. We understand a characteristic surface $\mathcal{C}\subset L$ as the enveloping surface of characteristic varieties $\CVL{e}$ as $e$ varies across $\mathcal C$, as the following definition clarifies.
\begin{definition}
\label{def: characteristics}
A surface $\mathcal{C}\subset L$ is called \textit{characteristic} if at any point $e\in\mathcal C$, the tangent space $T_e\mathcal C$ comprises one (and only one) light-like direction.
\end{definition}
We remark that Definition \ref{def: characteristics} closely parallels the notion of characteristics in general relativity, where they are identified with light-like surfaces \cite{EN2000}. Definition \ref{def: characteristics} may be considered as a straightforward generalization to nonlinear PDEs of the classical notion of characteristics for linear PDEs (see Appendix \ref{appendix: characteristics}). 
To motivate this statement, fix coordinates $\{q^1,q^2,q^3\}$ on $L$ and consider the surface
\begin{equation}
    \mathcal{C}=\{(q^1,q^2,q^3)\in L:F(q^1,q^2,q^3)=0\}.
\end{equation}
Further, consider the following vector based at points on $\mathcal{C}$,
\begin{equation}
    dF^\sharp=h^{ij}\frac{\partial F}{\partial q^i}\frac{\partial}{\partial q^j},
\end{equation}
where summation on repeated indices is implied and $h^{ij}$ denotes the components of the inverse metric $h_\omega^{-1}$. It is straightforward to check equivalence of the following statements: \textit{(i)} $dF^\sharp$ is a tangent vector to $\mathcal{C}$, \textit{(ii)} $dF^\sharp$ is a light-like vector, \textit{(iii)} $F$ satisfies
\begin{equation}
\label{eikonal L}
    h_\omega^{-1}(dF,dF)=0.
\end{equation}
Equation (\ref{eikonal L}) is the analogue of the eikonal equation (\ref{eikonal Rn}) in the linear setting, where the matrix coefficient is replaced by the inverse of $h_\omega$. We shall further elaborate on this analogy next. Note that equation (\ref{eikonal L}) is only well defined away from parabolic points, where $h_\omega$ is not invertible, and has only nontrivial solutions on hyperbolic branches of $L$. Any hyperbolic branch is regular by virtue of Proposition \ref{prop: Sigma L = parabolic}, so it may be described as the graph of a function, e.g. $P^*(x,y,z)$, to the extent a classical solution is (we may take $\{x,y,z\}$ as local coordinates on hyperbolic branches). Therefore, the linearization of (\ref{CS P}) about a hyperbolic branch of $L$ is well defined, and, building on equation (\ref{LR vs A}), we may write
\begin{equation}
    h_\omega^{-1}=\frac{1}{2\det(A)}A,
\end{equation}
where $A=\adj(\Hess(P^*))$ is the coefficient matrix of the linearized equation (\ref{CS P}) about $P^*$. Since $\det(A)\ne 0$ on hyperbolic points, we may get rid of this term in equation (\ref{eikonal L}), and write
\begin{equation}
    \nabla F\cdot A\nabla F=\sum_{i,j=1}^3 a_{ij}(x,y,z)\frac{\partial F}{\partial x^i}\frac{\partial F}{\partial x^j}=0,
\end{equation}
where $(x^1,x^2,x^3)=(x,y,z)$.
The analogy with (\ref{eikonal Rn}) should now be apparent.

Equation (\ref{eikonal L}) is a nonlinear PDE of the first order, and, as (\ref{eikonal Rn}), is solved by the classical methods of wave optics (see for example Appendix 4 of \cite{Arn78}). Namely, the solution surface is understood as foliated by light-like curves (i.e., the light rays of wave optics) which satisfy a set of Hamilton's canonical equations of motion. We briefly recall the main steps of the solution procedure for completeness of exposition \cite{Arn78,CH62}.
Consider the cotangent bundle $T^*L$ with coordinates 
\begin{equation}
    \{q^1,q^2,q^3,p_1,p_2,p_3\},
\end{equation}
and symplectic structure
\begin{equation}
    \Omega=dp_i\wedge\de q^i.
\end{equation}
Equation (\ref{eikonal L}) is thus interpreted as the zero level set of the Hamiltonian function $\mathcal{H}:T^*L\to\mathbb{R}$,
\begin{equation}
\label{Hamiltonian rays}
    \mathcal{H}(p,q)=h(q)^{ij}p_ip_j,
\end{equation}
under the identification $p_i=\partial F/\partial q^i$. Characteristic curves of (\ref{eikonal L}) are defined by \cite{Arn78} as the integral curves of the Hamiltonian vector field $\xi_\mathcal{H}$,
\begin{equation}
    -d\mathcal H=\iota_{\xi_\mathcal{H}}\Omega,
\end{equation}
and satisfy the Hamilton's canonical equations
\begin{equation}
\label{Hamilton's eq}
    \dot q^i=\frac{\partial \mathcal H}{\partial p_i}=2(h_\omega)^{ij}p_j, \qquad \dot p_i=-\frac{\partial \mathcal H}{\partial q^i}=-\frac{\partial (h_\omega)^{jk}}{\partial q^i}p_jp_k.
\end{equation}
Initial conditions are not free, but are subject to the condition
\begin{equation}
\label{ll rays}
    \mathcal H(p(0),q(0))=0.
\end{equation}
Integral curves of (\ref{Hamilton's eq}) are called \textit{bicharacteristics} \cite{Lyc85}. Once projected to $L$ along the cotangent bundle, $\Bar{\pi}:T^*L\to L$, bicharacteristics foliate the characteristic surfaces $\mathcal{C}\subset L$. The following commutative diagram summarizes the relations introduced so far.
\begin{equation*}
    \xymatrix{T^*L  \ar[d]_{\Bar{\pi}} & {}\\
    L \ar@{^{(}->}[r]^{\iota} \ar[dr]_{\pi_L} & \cotb{3} \ar[d]^{\pi}\\
    {} & \rn{3}
    }
\end{equation*}


\begin{remark}
\label{remark Lagrangian light rays}
Light rays are equivalently described by the Lagrangian
\begin{equation}
\label{Lagrangian rays}
    \mathcal{L}=h_{ij}(q)\dot q^i\dot q^j,
\end{equation}
related to the Hamiltonian (\ref{Hamiltonian rays}) by the classical Legendre transform.
The equivalent of (\ref{ll rays}) reads
\begin{equation}
\label{ll geodesics condition}
    \mathcal{L}(q(0),\dot q(0))=0.
\end{equation}
The Euler-Lagrange equations associated with (\ref{Lagrangian rays}) plus condition (\ref{ll geodesics condition}) precisely yield the light-like geodesics of $h_\omega$. Therefore, characteristic surfaces $\mathcal{C}$ are foliated by light-like geodesics of $h_\omega$, and this offers an alternative approach to computing them.
\end{remark}

In the next section we provide an example of a generalized solution to the SG equations, and show the interaction of characteristics, elliptic-hyperbolic transition and singularities.


\section{Exact solutions}
\label{sec:examples}
Not many exact solutions to the full semigeostrophic system (\ref{SG nondimensional}) are known \cite{BT97,Shu91}, and even fewer are the generalized ones. A common assumption often encountered in the literature is that of uniform potential vorticity, which helps finding particular solutions and possesses physical relevance. Under this assumption, equation (\ref{CS P}) decouples from system (\ref{SG nondimensional}), and is interpreted as a \MA equation for the unknown geopotential. The choice $\epsilon\PV=1$ brings (\ref{CS P}) to the form
\begin{equation}
\label{CS P 1}
    \hess(P)=1,
\end{equation}
widely studied in the literature. A 2-parameter family of classical solutions to (\ref{CS P 1}) is introduced in \cite{Ban2002} and generalized in \cite{LYZ2005}. Once a particular solution to (\ref{CS P 1}) is selected, one is able to build a full solution to (\ref{SG nondimensional}) as we show in Section \ref{sec:examples}.

\subsection{Construction of exact solutions}
We build generalized solutions to (\ref{CS P 1}) by solving (\ref{CS T}), which, under the assumption $\epsilon\PV=1$, takes the form
\begin{align}
    \label{CS T 1}
    &\frac{\partial^2T}{\partial x^2}\frac{\partial^2T}{\partial y^2}-\left(\frac{\partial^2T}{\partial x\partial y}\right)^2+\frac{\partial^2T}{\partial Z^2}=0.
\end{align}
Although time dependence is still possible for constant vorticity flows, we shall restrict to stationary solutions which nevertheless show some interesting features.
We look for analytical solutions to (\ref{CS T 1}) in the form
\begin{equation}
\label{expansion T}
    T=T^{(0)}(Z)+T^{(1)}_\alpha(Z)x^\alpha+\frac{1}{2}T^{(2)}_{\alpha\beta}(Z)x^\alpha x^\beta+\frac{1}{3!}T^{(3)}_{\alpha\beta\gamma}(Z)x^\alpha x^\beta x^\gamma+\dots
\end{equation}
where $(x^1,x^2)=(x,y)$ and summation on repeated indices is implied.
Several classes of finite dimensional reductions of (\ref{expansion T}) are possible. Third order truncation of the above series provides a wide class of fully polynomial solutions whose coefficients satisfy
\begin{equation}
\label{T3}
    \frac{d^2\TTT{3}{\alpha\beta\gamma}}{dZ^2}=0,
\end{equation}
\begin{equation}
\label{T2}
    \begin{cases}
    \displaystyle{\frac{d^2\TTT{2}{11}}{dZ^2}+2\det(\TTT{3}{1\alpha\beta})=0},\\
    \displaystyle{\frac{d^2\TTT{2}{12}}{dZ^2}+\TTT{3}{111}\TTT{3}{222}-\TTT{3}{112}\TTT{3}{122}},\\
    \displaystyle{\frac{d^2\TTT{2}{22}}{dZ^2}+2\det(\TTT{3}{2\alpha\beta})=0},
    \end{cases}
\end{equation}
\begin{equation}
    \label{T1}
    \begin{cases}
    \displaystyle{\frac{d^2\TTT{1}{1}}{dZ^2}+\TTT{2}{22}\TTT{3}{111}-2\TTT{2}{12}\TTT{3}{112}+\TTT{2}{11}\TTT{3}{222}=0},\\
    \displaystyle{\frac{d^2\TTT{1}{2}}{dZ^2}+\TTT{2}{22}\TTT{3}{112}-2\TTT{2}{12}\TTT{3}{122}+\TTT{2}{11}\TTT{3}{222}=0},
    \end{cases}
\end{equation}
\begin{equation}
    \label{T0}
    \frac{d^2\TTT{0}{}}{dZ^2}+\det(\TTT{2}{\alpha\beta})=0.
\end{equation}
A straightforward inspection of equations (\ref{T3})--(\ref{T0}) shows that the coefficient functions are polynomials of a definite degree in $Z$. Specifically, $\TTT{3}{\alpha\beta\gamma}$ has degree 1, $\TTT{2}{\alpha\beta}$ has degree 4, $\TTT{1}{\alpha}$ has degree 7, and $\TTT{0}{}$ has degree 10. 
Retaining terms of the fourth order and higher in the expansion (\ref{expansion T}) leads to a wider class of exact solutions, though they are generally not polynomial in the vertical variable.


Polynomial solutions are particularly valuable for their ability to encode the local behaviour of more complex solutions, and, in particular, the singular structure. The simplest nontrivial Lagrangian singularity is the fold ($A_2$), and the germ of a Lagrangian submanifold with this feature is canonically described by a generating function
\begin{equation}
\label{canonical form ex0}
    T^*(x,y,Z)=\frac{Z^3}{6}
\end{equation}
(see for example \cite{AGV2012} for a list of low-dimensional canonical forms of elementary catastrophes).
We build an example solution to (\ref{CS T 1}) by deformation of (\ref{canonical form ex0}) through the addition of a polynomial term. 
One of the simplest choices is
\begin{figure}
    \centering
    \includegraphics[width=0.25\textwidth]{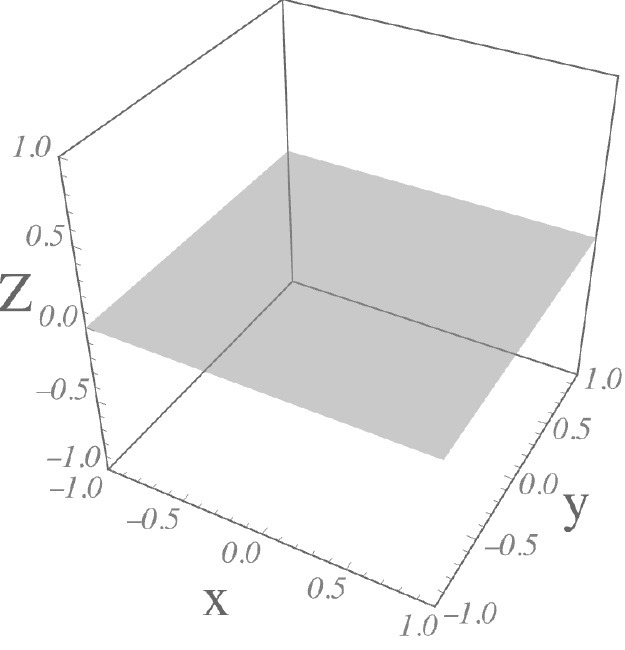}
    \hspace{30pt}
    \includegraphics[width=0.25\textwidth]{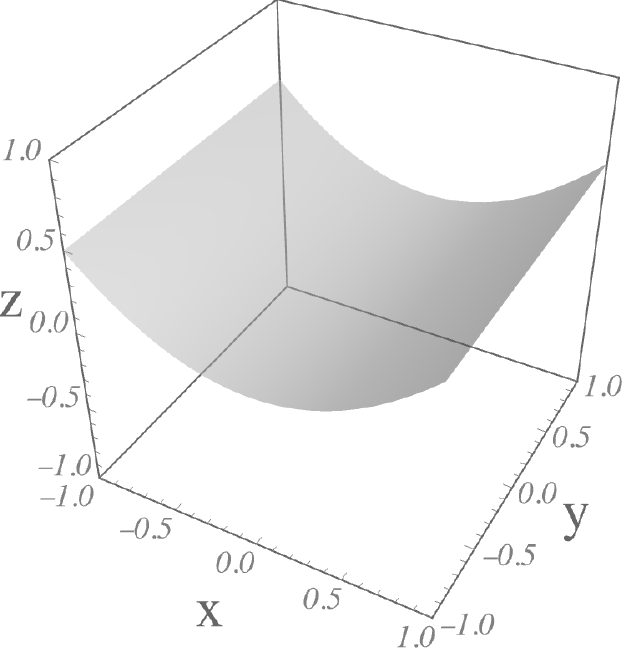}
    \caption{Singular locus and caustics of the generalized solution (\ref{L ex0}) showing the typical appearance of the $A_2$ singularity.}
    \label{fig:ex0_sig_cau}
\end{figure}
\begin{equation}
\label{T ex0}
    T=\frac{y^2}{2}-\frac{x^2 Z}{2}+\frac{Z^3}{6}.
\end{equation}

\subsection{Lagrangian submanifold and projection}
The generating function (\ref{T ex0}) determines a Lagrangian submanifold according to
\begin{align}
\label{L ex0}
    L=\bigg\{(x,y,z,X,Y,Z)\in\cotb{3}:X=\frac{\partial T}{\partial x}=xZ,Y=\frac{\partial T}{\partial y}=y,\\
    z=-\frac{\partial T}{\partial Z}=\frac{x^2}{2}-\frac{Z^2}{2}\bigg\}.
\end{align}
The restriction $\pi_L:=\pi|_L$ of the bundle projection (\ref{Lagrangian projection L}) to $L$ in local coordinates is
\begin{equation}
\label{piL ex0}
    \pi_L(x,y,Z)=\left(x,y,-\frac{\partial T}{\partial Z}\right)=\left(x,y,\frac{x^2}{2}-\frac{Z^2}{2}\right),
\end{equation}
and shows that the fold singularity occurs across the plane
\begin{equation}
\label{cau ex0}
    \Sigma L=\left\{(x,y,Z)\in L: 0=\det(\de\pi_L)=-\frac{\partial^2 T}{\partial Z^2}=-Z\right\}.
\end{equation}
The projection of the singular locus to the physical space identifies the caustics,
\begin{equation}
    \pi(\Sigma L)=\left\{(x,y,z)\in\rn{3}: z=\frac{x^2}{2}\right\},
\end{equation}
and Figure \ref{fig:ex0_sig_cau} provides a view of them.
Due to the nature itself of the fold singularity, the caustics bound the solution domain in the physical space. In other words, the solution is only defined in the domain $z\le x^2/2$ of $\rn{3}(x,y,z)$.
\begin{remark}
A different approach is used in \cite{CS89}, where the authors put the focus on the singularities of the geopotential graph.
The \CS relations (\ref{Legendre transform R})--(\ref{Legendre transform T}) yield the graph of the multivalued geopotential $P(x,y,z)$ as parametrized by $Z$,
\begin{equation}
\label{P ex0}
    \begin{cases}
    \displaystyle{P=\frac{y^2}{2}-\frac{Z^3}{3}},\\
    \displaystyle{z=\frac{x^2}{2}-\frac{Z^2}{2}},
    \end{cases}
\end{equation}
and Figure \ref{fig:ex0_P} shows a section of it for $y=constant$. The cusped edge is a distinctive feature of the $A_2$ singularity in the \textit{Legendrian} context.
\begin{figure}
    \centering
    \includegraphics[width=0.3\textwidth]{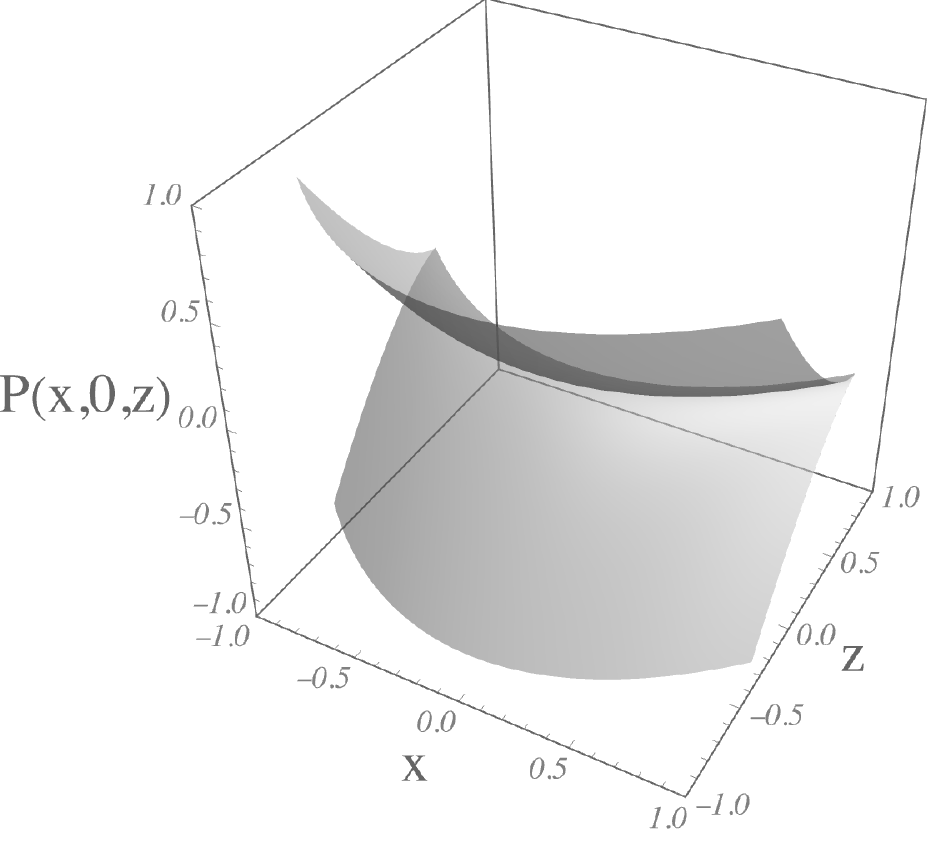}
    \caption{A slice through the graph of the multivalued geopotential (\ref{P ex0}) for $y=0$.}
    \label{fig:ex0_P}
\end{figure}
\end{remark}

\subsection{\LR metric}
Next we examine the pseudo-Riemannian geometry of the solution. The \LR metric on $L$ in local coordinates is
\begin{align}
\label{LR metric ex0}
\begin{split}
    h_\omega&=2(T_{xx}\de x^2+2T_{xy}\de x\de y+T_{yy}\de y^2-T_{ZZ}\de Z^2)\\
    &=2(-Z\de x^2+\de y^2-Z\de Z^2).
    \end{split}
\end{align}
This immediately implies, according to Definition \ref{def: type vs h}, that the problem is elliptic for $Z<0$ and hyperbolic for $Z>0$. Characteristic surfaces in the hyperbolic region are determined by the light-like geodesics of (\ref{LR metric ex0}) (see Remark \ref{remark Lagrangian light rays}). General geodesics satisfy
\begin{equation}
\label{geodesics ex0}
    \ddot x=-\frac{\dot x\dot Z}{Z},\qquad \ddot y=0,\qquad \ddot Z=\frac{\dot x^2-\dot Z^2}{2Z}.
\end{equation}
The first two equations can be immediately integrated once, and yield two of constants of the motion,
\begin{equation}
\label{constants ex0}
    \dot x Z=C_1,\qquad \dot y=C_2.
\end{equation}
Using (\ref{constants ex0}) in the third of (\ref{geodesics ex0}), the problem of finding general geodesics is reduced to the single equation
\begin{equation}
\label{geodesics ex0 Z}
    2Z^3\ddot Z=C_1^2-Z^2\dot Z^2.
\end{equation}
A further simplification is available for
light-like geodesics, which are subject to the additional constraint
\begin{equation}
\label{ll geodesics}
    -Z\dot x^2+\dot y^2-Z\dot Z^2=0.
\end{equation}
Using (\ref{constants ex0}) in (\ref{ll geodesics}) gives the separable equation
\begin{equation}
\label{ll geodesics Z}
    C_1^2-C_2^2Z+Z^2\dot Z^2=0,
\end{equation}
with implicit solution,
\begin{equation}
\label{implicit solution ex0}
    \pm s=\frac{2\sqrt{C_2^2 Z-C_1^2}(2C_1^2+C_2^2 Z)}{3C_2^4}-\frac{2\sqrt{C_2^2 Z_0-C_1^2}(2C_1^2+C_2^2 Z_0)}{3C_2^4}.
\end{equation}
We shall remark at this point that (\ref{ll geodesics}) is compatible with the geodesics equations (\ref{geodesics ex0}), as can be verified by taking a derivative of (\ref{ll geodesics}) and using (\ref{geodesics ex0}) to eliminate the second derivative terms. Indeed, any solution of (\ref{ll geodesics Z}) is a geodesic curve. This can be checked by taking the derivative of (\ref{ll geodesics Z}) with respect to the curve parameter to get
\begin{equation}
    -C_2^2+2Z\dot Z^2+2Z^2\ddot Z=0.
\end{equation}
Then, multiplication by $Z$ and the use of (\ref{ll geodesics}) to eliminate $C_2$ leads back to equation (\ref{geodesics ex0 Z}). Next, we use the first constant of the motion in (\ref{constants ex0}) to get $x(s)$,
\begin{equation}
\label{x(Z) ex0}
    \frac{dx}{dZ}=\frac{C_1}{Z}\frac{ds}{dZ}=\pm\frac{C_1}{\sqrt{C_2^2Z-C_1^2}}.
\end{equation}
Once integrated, equation (\ref{x(Z) ex0}) plus (\ref{implicit solution ex0}) and the second of (\ref{constants ex0}) allows us to write the light-like geodesics as
\begin{equation}
    \begin{cases}
    x\pm\frac{2\dot x_0Z_0}{\dot y_0^2}\sqrt{\dot y_0^2 Z-\dot x_0^2Z_0^2}=constant,\\
    y\pm \frac{2}{3\dot y_0^3}(2\dot x_0^2Z_0^2+\dot y_0^2Z)\sqrt{\dot y_0^2Z-\dot x_0^2Z_0^2}=constant,
    \end{cases}
\end{equation}
where we have used $C_1=\dot x_0Z_0$, $C_2=\dot y_0$.
To see how characteristics interact with the parabolic boundary $\Sigma L$, we assume next that the geodesic starting point $(x_0,y_0,Z_0)$ belongs to $\Sigma L$, that is $Z_0=0$. This results in
\begin{equation}
    \begin{cases}
    x=x_0,\\
    (y-y_0)^2=\frac{4}{9}Z^3.
    \end{cases}
\end{equation}
Thus, characteristics intersecting the parabolic boundary form semicubical cusps at the intersection point, as expected from the literature on the subject \cite{LL87}. Figure \ref{fig:ex0_cha} provides a view of the characteristic surfaces near the singular locus.
\begin{figure}
    \centering
    \includegraphics[width=0.3\textwidth]{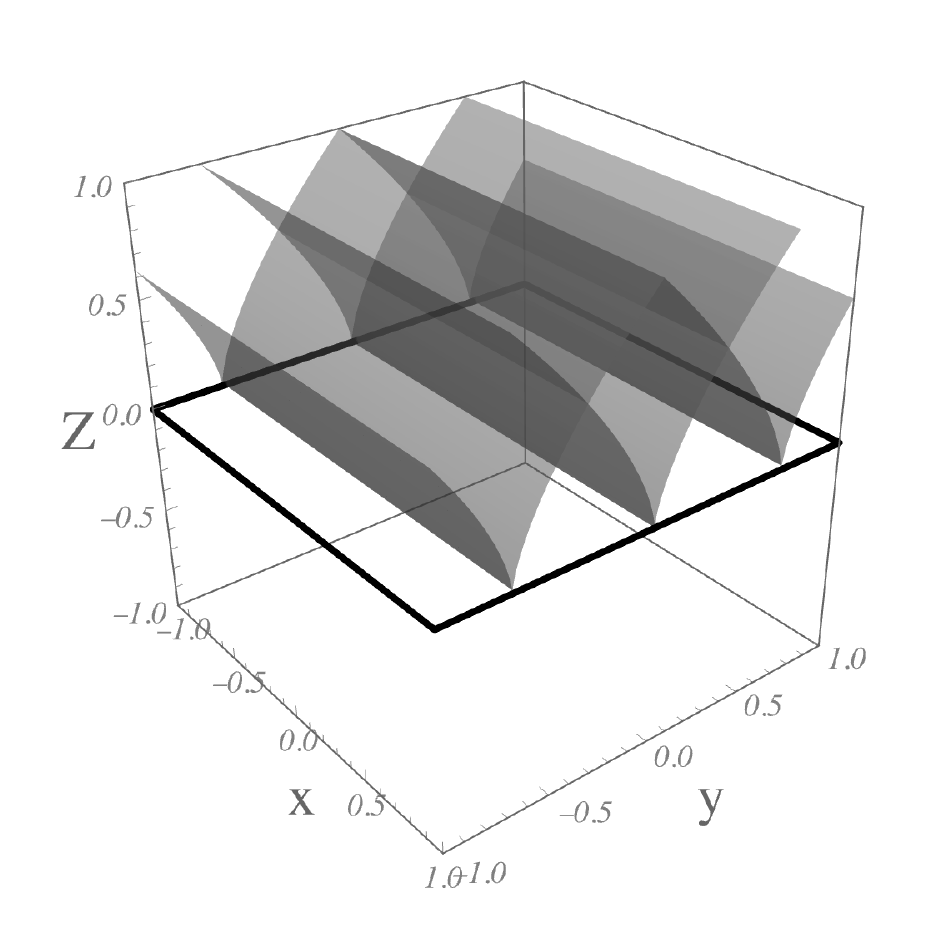}
    \caption{Characteristic surfaces (gray) intersecting the boundary of the hyperbolic region, $Z=0$. Intersection points are found along parallel lines along the parabolic plane $Z=0$ where the characteristics form semicubical cusps.}
    \label{fig:ex0_cha}
\end{figure}

\subsection{Full semigeostrophic solution}
In the remainder of this section we explicitly reconstruct the full semigeostrophic solution from the knowledge of the generating function (\ref{T ex0}). Thanks to the simple structure of the solution (\ref{T ex0}), we are able to explicitly (piecewise) invert the relation (\ref{P ex0}) and write
\begin{equation}
\label{P ex0 xyz}
    \epsilon\theta=Z=\pm\sqrt{x^2-2z},\qquad 
    P=\frac{y^2}{2}\mp\frac{1}{3}(x^2-2z)^{3/2}.
\end{equation}
Once the geopotential $P$ is known, absolute momentum and potential temperature are obtained by derivation as
\begin{equation}
    M=\frac{\partial P}{\partial x}=\mp x\sqrt{x^2-2z},\qquad N=\frac{\partial P}{\partial y}=y.
\end{equation}
Next, the geostrophic wind is found as
\begin{equation}
    u_g=\PV(y-N)=0,\qquad v_g=\PV(M-x)=\PV(\mp x\sqrt{x^2-2z}-x),
\end{equation}
where we have used $\epsilon\PV=1$.
The momentum balance equations plus the transport of potential temperature yield a system of algebraic equations for the unknown components of the velocity field,
\begin{equation}
\label{determining system u v w ex0}
    \begin{cases}
    M_x u+M_y v+M_z w=u_g,\\
    N_x u+N_y v+N_z w=v_g,\\
    \theta_x u+\theta_y v+\theta_z w=0.
    \end{cases}
\end{equation}
Since $N=y$, it easily follows that $v=v_g$. Moreover, $M_y=\theta_y=0$ and $u_g=0$, which imply that $u=w=0$. Indeed, the first and the last equations in (\ref{determining system u v w ex0}) form a linear homogeneous system with nondegenerate coefficient matrix as
\begin{equation}
    \frac{\partial(M,\theta)}{\partial(x,y)}=\hess(P)=1.
\end{equation}
In conclusion, the flow field corresponding to (\ref{T ex0}) is a purely geostrophic meridional wind,
\begin{equation}
\label{velocity field ex0}
    u=0,\qquad v=v_g=\PV(\mp x\sqrt{x^2-2z}-x),\qquad w=0.
\end{equation}
To restore single-valuedness of the solution, Chynoweth and Sewell appealed to the convexity principle of \cite{SC87}. Namely, only convex branches of the multivalued graph of $P$ are retained while concave ones are discarded. Although the application of this principle requires some attention in the general case, it is straightforward in this example. The admissible branch of $P$ is found to be,
\begin{equation}
    P=\frac{y^2}{2}+\frac{1}{3}(x^2-2z)^{3/2},
\end{equation}
which corresponds to the elliptic branch of the multivalued solution, $Z<0$.
This corresponds to the velocity field
\begin{equation}
\label{velocity ex0}
    u=0,\qquad v=v_g=\PV(x\sqrt{x^2-2z}-x),\qquad w=0,
\end{equation}
and represents a geostrophic wind in the northern hemisphere directed poleward. Figure \ref{fig:velocity ex0} shows the wind magnitude on a section normal to the flow.

This example shares qualitative features with a larger class of exact solutions, i.e., 2-dimensional stationary flows. These flows are characterized by the independence of the geopotential $\phi$ of one of the horizontal coordinates (in this case $y$) which results in a vanishing zonal component of the geostrophic wind. Under stationary conditions, flows in this class are purely geostrophic (either zonal or meridional).
\begin{figure}
    \centering
    \includegraphics[width=0.4\textwidth]{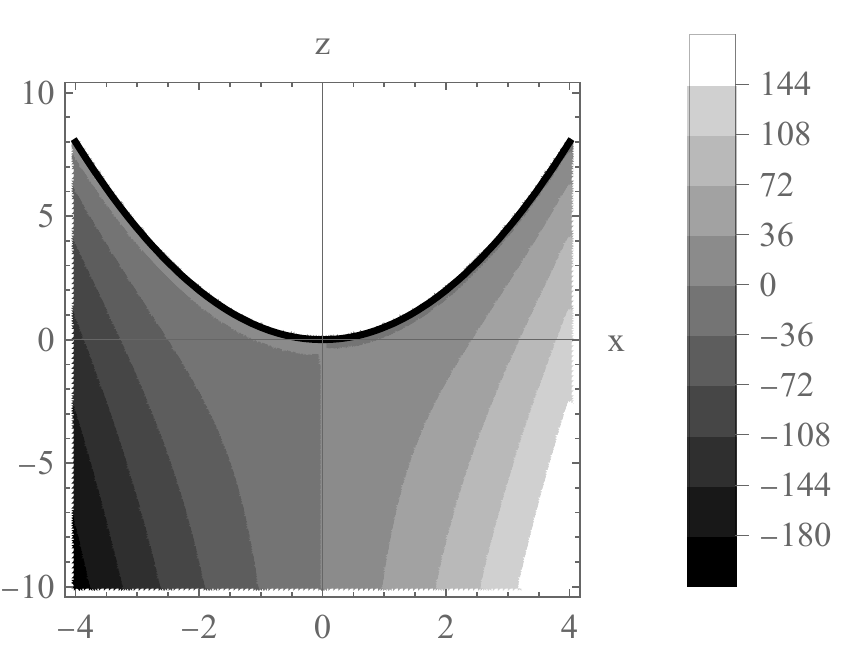}
    \caption{Magnitude of the geostrophic wind (\ref{velocity field ex0}) on a section normal to the flow. The fluid domain is bounded from above by the caustics (\ref{cau ex0}).}
    \label{fig:velocity ex0}
\end{figure}

\section{Conclusions and future directions}
The \LR metric has been much studied in the context of \MA geometry, but its pull-back to generalized solutions, realized as Lagrangian submanifolds, has hitherto been unexplored. We have explored this feature from the point of view of PDE theory in the physically and mathematically important example of the semigeostrophic equations. 

In particular, we have shown connections between the signature of the pull-back metric on solutions, the symbol type of the \MA equation, and its role in describing the characteristic surfaces of hyperbolic equations. We recognise the pull-back metric as a tool for studying singularities,  which complements and extends the earlier work Kossowski \cite{Kos91}, where a version of the \LR metric on $\cotb{2}$ was the primary object of interest. 

Several questions are still open. We illuminated the meaning of the light-like geodesics in hyperbolic regime, but the potential role of space-like and time-like geodesics, and the elliptic regime in this context, remain to be explored.

Another intriguing question is the geometrical and physical meaning of the curvature of the Lagrangian submanifolds, and its relationship with singularities. This aspect has been explored in the work of Napper \textit{et al.} on Navier-Stokes equations, and its implications for semigeostrophic theory is matter for future research. We focused on the kinematic aspects of the SG equations, considering time as a fixed parameter, but the system dynamics is important. Considering time-dependent solutions leads to a 1-parameter family of metrics, i.e., a notional geometric flow, whose properties are unknown. We might speculate a relation between such a geometric flow and the onset of dynamic singularities.

\appendix


\section*{Acknowledgements}
We would like to thank T. Bridges, L. Napper, and M. Wolf for many useful discussions.
R.D. and G.O. were supported by the European Union’s Horizon 2020 research and innovation program under the Marie Skłodowska-Curie grant no 778010 IPaDEGAN. R.D. and G.O. thank the financial support of the
project MMNLP (Mathematical Methods in Non Linear Physics) of the
INFN.
R.D. and G.O. also gratefully acknowledge the auspices of the GNFM Section of INdAM under which part of this work was carried out. This work is part of R.D.'s dual PhD program Bicocca-Surrey.

\section*{Data and Licence Management}

No additional research data beyond the data presented and cited in this work are needed to validate the research findings in this work. For the purpose of open access, the authors have applied a Creative Commons Attribution (CC BY) licence to any Author Accepted Manuscript version arising.

\section{Characteristics of linear PDEs}
\label{appendix: characteristics}
We recall here some classical terminology from the theory of linear PDEs (see for example \cite{CH62}). Let a second order linear PDE in $n$ independent variables $x=(x_1,\dots,x_n)$ have principal part
\begin{equation}
    \sum_{i,j=1}^n a_{ij}(x)\frac{\partial^2 u}{\partial x_i\partial x_j}.
\end{equation}
The $x$-depending quadratic form
\begin{equation}
    \sigma(x,\mathbf{\xi})=\sum_{i,j=1}^n a_{ij}(x)\xi_i\xi_j,\qquad \mathbf{\xi}:=(\xi_1,\dots,\xi_n),
\end{equation}
is called the \textit{principal symbol} of the equation. A vector $\mathbf\xi$ based at $x$ is called \textit{characteristic} if $\sigma(x,\mathbf{\xi})=0$ and the set of characteristic vectors at $x$ is called the \textit{Fresnel cone}. An implicitly defined hypersurface $F(x_1,\dots,x_n)=0$ is called \textit{characteristic surface} or simply \textit{characteristic} if its normal vector is characteristic. In other words, $F$ satisfies the \textit{eikonal} equation
\begin{equation}
\label{eikonal Rn}
    \sigma(x,\nabla F)=\sum_{i,j=1}^n a_{ij}(x)\frac{\partial F}{\partial x_i}\frac{\partial F}{\partial x_j}=0.
\end{equation}


\begin{thebibliography}{9}

\bibitem{Arn78}
V. I. Arnold (1978) \textit{Mathematical Methods of Classical Mechanics}, Springer, New York




\bibitem{AGV2012}
V. I. Arnold, S. M. Gusein-Zade and A. N. Varchenko (2012) \textit{Singularities of Differentiable Maps, Volume 1, Classification of Critical Points, Caustics and Wave Fronts}, Birkhäuser Basel

\bibitem{Ban2002}
B. Banos (2002) ``Nondegenerate Monge–Ampère Structures in Dimension 6'', \textit{Letters in Mathematical Physics}, \textbf{62}, 1--15


\bibitem{BT97}
H. R. Birkett \& A. J. Thorpe (1997) ``Superposing semi-geostrophic potential-vorticity anomalies'', \textit{Q. J. R. Meteorol. Soc.}, \textbf{123}(543), 2157--2163

\bibitem{BRR2016}
B. Banos, V. N. Roubtsov \& I. Roulstone (2016) ``Monge–Ampère structures and the geometry of incompressible flows'', \textit{J. Phys. A: Math. Theor.}, \textbf{49}, 244003



\bibitem{CS89}
S. Chynoweth \& M. J. Sewell (1989) Dual variables in semigeostrophic theory, \textit{Proc. R. Soc. Lond. A}, \textbf{424}, 155--186




\bibitem{CH62}
R. Courant \& D. Hilbert (1962) \textit{Methods of mathematical physics. / Volume II, Partial differential equations}, Wiley Interscience, New York

\bibitem{CP84}
M. J. P. Cullen \& R. J. Purser (1984) ``An extended Lagrangian theory of semigeostrophic frontogenesis'', \textit{J. Atmos.
Sci.}, \textbf{41}, 1477--97


\bibitem{CR93}
M. J. P. Cullen \& I. Roulstone (1993) ``A Geometric Model of the Nonlinear Equilibration of Two-Dimensional Eady Waves'', \textit{J. Atmos. Sci}, \textbf{50}(2), 328--332

\bibitem{DR2010}
S. Delahaies \& I. Roulstone (2010) ``Hyper-Kähler geometry and semi-geostrophic theory'', \textit{Proc. R. Soc. A.}, \textbf{466}, 195--211



\bibitem{Duz2004}
S.V. Duzhin (2004) ``Infinitesimal Classification of Systems of Two First Order Partial Differential Equations in Two Variables'', \textit{J. Math. Sci}, \textbf{119}, 30--34

\bibitem{EN2000}
J. Ehlers \& E. T. Newmann (2000) ``The theory of caustics and wave front singularities with physical applications'', \textit{J. Math. Phys.}, \textbf{41}, 3344



\bibitem{HL82}
R. Harvey \& H. Lawson (1982) ``Calibrated geometry'', \textit{Acta Mathematica}, \textbf{148}, 47--157

\bibitem{HS90}
M. W. Holt \& G. J. Shutts (1990) ``An analytical model of the growth of a frontal discontinuity'', \textit{Q. J. R. Meteorol. Soc.}, \textbf{116}, 269--286

\bibitem{HB72}
B. J. Hoskins \& F. P. Bretherton (1972) ``Atmospheric Frontogenesis Models: Mathematical Formulation and Solution'', \textit{J. Atmos. Sci.}, \textbf{29}, 11--37

\bibitem{Hos75}
B. J. Hoskins (1975) ``The Geostrophic Momentum Approximation and the Semi-Geostrophic Equations'', \textit{J. Atmos. Sci.}, \textbf{32}, 233--242


\bibitem{IM2006a}
G. Ishikawa \& Y. Machida (2006) ``Singularities of improper affine spheres and surfaces of
constant Gaussian curvature'', \textit{International J. Math.}, \textbf{17}(3), 269--293


\bibitem{IM2006b}
G. Ishikawa and Y. Machida (2006) ``Extra singularities of geometric solutions to Monge-Ampère equation of three variables'', \textit{Kyoto Univ. Res. Inf. Repos.}, \textbf{1502}, 41--53


\bibitem{Kos91}
M. Kossowski (1991) ``Local Existence of Multivalued Solutions to Analytic Symplectic Monge-Ampère Equations (The Nondegenerate and Type Changing Cases)'', \textit{Indiana University Mathematics Journal}, \textbf{40}(1), 123--148

\bibitem{KLR2006}
A. Kushner, V. Lychagin \& V. Rubtsov (2006) \textit{Contact Geometry and Nonlinear Differential Equations}, Cambridge University Press

\bibitem{LL87}
L. D. Landau \& E. M. Lifshitz (1987) \textit{Fluid Mechanics -- Course of Theoretical Physics, Volume 6}, 2nd Edition, Pergamon Press


\bibitem{Lyc85}
V. Lychagin (1985) ``Singularities of multivalued solutions of nonlinear differential equations, and nonlinear phenomena'', \textit{Acta Appl. Math.}, \textbf{3}, 135--173

\bibitem{LR83}
V. V. Lychagin \& V. N. Rubtsov (1983) ``Local classification of Monge–Amp`ere differential equations'', \textit{Dokl. Akad. Nauk SSSR}, \textbf{272}(1), 34--38





\bibitem{LYZ2005}
J. Loftin, S.-T. Yau \& E. Zaslow (2005) ``Affine manifolds, SYZ geometry and the “Y” vertex'', \textit{J. Differential Geom.}, \textbf{71}(1), 129--158






\bibitem{NRR2022}
L. Napper, I. Roulstone, V. Rubtsov \& M. Wolf (2022) ``\MA geometry and vortices'' (to appear)




\bibitem{Oli2006}
M. Oliver (2006) ``Variational asymptotics for rotating shallow water near geostrophy: a transformational approach'', \textit{J. Fluid Mech.}, \textbf{551}, 197--234


\bibitem{RBG2009}
I. Roulstone, B. Banos, J. D. Gibbon \& V. N. Roubtsov (2009) ``A Geometric Interpretation of Coherent Structures in Navier-Stokes Flows'', \textit{Proceedings: Mathematical, Physical and Engineering Sciences}, \textbf{465}(2107), 2015--2021



\bibitem{RN2006}
I. Roulstone \& J. Norbury (2006) ``A Hamiltonian structure with contact geometry for the semi-geostrophic equations'', \textit{J. Fluid Mech.}, \textbf{272}, 211--234



\bibitem{Shu91}
G. Shutts (1991) ``Some exact solutions to the semi-geostrophic equations for uniform potential vorticity flows'', \textit{Geophys. Astrophys. Fluid Dyn.}, \textbf{57}(1-4), 99--114

\bibitem{SC87}
G. J. Shutts \& M. J. P. Cullen (1987) ``Parcel Stability and its Relation to Semigeostrophic Theory'', \textit{J. Atmos. Sci}, \textbf{44}(9), 1318--1330

\bibitem{VK77}
A. M. Vinogradov \& B. A. Kupershmidt (1977) ``The Structures of Hamiltonian Mechanics'', \textit{Russ. Math. Surv.}, \textbf{32}, 177

\end{thebibliography}
\end{document}